\documentclass[aps,amsmath,onecolumn,amssymb,notitlepage]{revtex4-1}
\usepackage{amssymb}
\usepackage{amsthm}
\usepackage{amsfonts}
\usepackage{complexity}
\usepackage{graphicx}
\usepackage{dcolumn}
\usepackage{bm}
\usepackage{hyperref}
\usepackage{enumerate}
\usepackage{algorithm}
\usepackage{algpseudocode}

\usepackage{multirow}

\newtheorem{theorem}{Theorem}
\newtheorem{lemma}{Lemma}
\newtheorem{definition}{Definition}
\newtheorem{corollary}{Corollary}

\usepackage[usenames,dvipsnames]{xcolor}
\hypersetup{
    colorlinks=true,       
    linkcolor=Maroon,          
    citecolor=OliveGreen,        
    filecolor=magenta,      
    urlcolor=Blue           
}

\begin{document}

\def\ket#1{\left|#1\right\rangle}
\def\bra#1{\langle#1|}
\newcommand{\ketbra}[2]{|#1\rangle\!\langle#2|}
\newcommand{\braket}[2]{\langle#1|#2\rangle}
\newcommand{\prob}[1]{{\rm Pr}\left(#1 \right)}
\newcommand{\expect}[2]{{\mathbb{E}_{#2}}\!\left\{#1 \right\}}
\newcommand{\var}[2]{{\mathbb{V}_{#2}}\!\left\{#1 \right\}}

\newcommand{\sinc}{\operatorname{sinc}}


\newcommand{\sde}{\mathrm{sde}}
\newcommand{\Z}{\mathbb{Z}}
\newcommand{\RR}{\mathbb{R}}
\newcommand{\w}{\omega}
\newcommand{\Kap}{\kappa}

\newcommand{\Tchar}{$T$}
\newcommand{\T}{\Tchar~}
\newcommand{\TT}{\mathrm{T}}
\newcommand{\ClT}{\{{\rm Clifford}, \Tchar\}~}
\newcommand{\Tcount}{\Tchar--count~}
\newcommand{\Tcountper}{\Tchar--count}
\newcommand{\Tcounts}{\Tchar--counts~}
\newcommand{\Tdepth}{\Tchar--depth~}
\newcommand{\Zr}{\Z[i,1/\sqrt{2}]}
\newcommand{\ve}{\varepsilon}

\newcommand{\eq}[1]{\hyperref[eq:#1]{(\ref*{eq:#1})}}
\renewcommand{\sec}[1]{\hyperref[sec:#1]{Section~\ref*{sec:#1}}}
\newcommand{\app}[1]{\hyperref[app:#1]{Appendix~\ref*{app:#1}}}
\newcommand{\fig}[1]{\hyperref[fig:#1]{Figure~\ref*{fig:#1}}}
\newcommand{\thm}[1]{\hyperref[thm:#1]{Theorem~\ref*{thm:#1}}}
\newcommand{\lem}[1]{\hyperref[lem:#1]{Lemma~\ref*{lem:#1}}}
\newcommand{\tab}[1]{\hyperref[tab:#1]{Table~\ref*{tab:#1}}}
\newcommand{\cor}[1]{\hyperref[cor:#1]{Corollary~\ref*{cor:#1}}}
\newcommand{\alg}[1]{\hyperref[alg:#1]{Algorithm~\ref*{alg:#1}}}
\newcommand{\defn}[1]{\hyperref[def:#1]{Definition~\ref*{def:#1}}}

\newcommand{\targfix}{\qw {\xy {<0em,0em> \ar @{ - } +<.39em,0em>
\ar @{ - } -<.39em,0em> \ar @{ - } +
<0em,.39em> \ar @{ - }
-<0em,.39em>},<0em,0em>*{\rule{.01em}{.01em}}*+<.8em>\frm{o}
\endxy}}

\newenvironment{proofof}[1]{\begin{trivlist}\item[]{\flushleft\it
Proof of~#1.}}
{\qed\end{trivlist}}

\newcommand{\cu}[1]{{\textcolor{red}{#1}}}
\newcommand{\tout}[1]{{}}
\newcommand{\good}{{\rm good}}
\newcommand{\bad}{{\rm bad}}

\newcommand{\id}{\openone}

\title{Can small quantum systems learn?}
\author{Nathan Wiebe$^\dagger$}
\author{Christopher Granade$^{*+}$}
\affiliation{$^\dagger$Quantum Architectures and Computation Group, Microsoft Research, Redmond, WA (USA)}
\affiliation{
    $^*$Centre for Engineered Quantum Systems,
    University of Sydney,
    Sydney, NSW, Australia
}
\affiliation{
    $^+$School of Physics,
    University of Sydney,
    Sydney, NSW, Australia
}
\begin{abstract}
We examine the question of whether quantum mechanics places limitations on the ability of small quantum devices to learn.  We specifically examine the question in the context of Bayesian inference, wherein the prior and posterior distributions are encoded in the quantum state vector.  We conclude based on lower bounds from Grover's search that an efficient blackbox method for updating the distribution is impossible.  We then address this by providing a new adaptive form of approximate quantum Bayesian inference that is polynomially faster than its classical anolog and tractable if the quantum system is augmented with classical memory or if the low--order moments of the distribution are protected using a repetition code.  This work suggests that there may be a connection between fault tolerance and the capacity of a quantum system to learn from its surroundings.
\end{abstract}
\maketitle
\section{Introduction}


Quantum systems have, at first glance, an incredible capacity to store vectors.  Only $n$ qubits suffice to represent a vector in $\mathbb{R}^{2^n}$.  This, in part, is the origin of all of the celebrated exponential speedups that quantum computing offers~\cite{shor1994algorithms,lloyd1996universal,childs2003exponential,HHL09}.
At the same time, from a computational learning perspective, quantum states can be approximately described
much more efficiently \cite{aaronson_learnability_2007}. 
Indeed, the tension between this ability of a small quantum register to store a very high--dimensional vector and the perspective offered by learning raises an important question: can a quantum agent with a logarithmically sized quantum memory efficiently learn from its surroundings?
Here, we seek to shed light on this question by investigating Bayesian inference as a way to model the quantum agent inferring properties about its surroundings.  We will see that lower bounds on quantum query complexity of unstructured search place severe limitations on the ability of the system to learn exactly.  Nonetheless, we will also see that quantum mechanics reveals new possibilities for approximate learning that are not present in classical Bayesian inference.

Bayes' rule is the heart of Bayesian inference.  It gives the correct way to update a prior distribution that describes the users' initial beliefs about a system model when a piece of experimental evidence is received.  If $E$ is a piece of evidence (an observable variable) and $x$ denotes a candidate model for the experimental system (a latent or hidden variable), then Bayes' rule states that the probability that the model is valid given the evidence (denoted $P(x|E)$) is 
\begin{equation}
P(x|E)= \frac{P(E|x) P(x)}{P(E)}= \frac{P(E|x) P(x)}{\langle P(E|x), P(x) \rangle},\label{eq:Bayes}
\end{equation}
where $P(E|x)$ is known as the likelihood function and is assumed to either be either numerically or inferred emperically.
This form of learning has a number of advantages for applications in quantum information processing \cite{granade_robust_2012,ferrie_likelihood-free_2014}.  Firstly, it is highly robust to noise and experimental imperfections.  Secondly, it is broadly applicable to almost every data processing problem.  Finally, Bayesian inference provides a probability distribution rather than a point estimate of the model. This allows the uncertainty in the inferred parameter to be directly computed.

There are several features that can make Bayesian inference computationally expensive, especially for scientific applications.  Perhaps the biggest contributor to the cost of these algorithms is the dimensionality of the model space.  Typically the latent variable $x$ is parameterized by a vector in $\mathbb{R}^d$, which means that precisely performing an update requires integrating over an infinite number of hypotheses.  Such integrals are often intractable, which limits the applicability of exact Bayesian inference.  

 A natural question to ask at this point is whether quantum computing could make inference tractable.  Quantum advantages are seen for a wealth of other machine learning protocols~\cite{aimeur2006machine,wiebe_quantum_2014,wiebe2014quantum,lloyd2013quantum,lloyd2014quantum,boixo2015fast}, so it stands to reason that it may be able to provide advantages here as well.  This issue is has been recently discussed in~\cite{LYC14}, which uses ideas from quantum rejection sampling~\cite{HHL09,ORR13} to accelerate the inference process.  Their work leaves a number of important issues open.  The method has success probability that shrinks exponentially with the number of updates attempted.  Furthermore, the algorithm cannot be applied in an online fashion nor can it be applied to continuous problems or those with stochastically varying latent variables.  We address these issues here by providing a quantum algorithm that can implement Bayesian inference in an online fashion by periodically classically caching a model of the posterior distribution.  This approach allows the quantum agent to revert to a previous model in the event that the update process fails, which allows the process to be implemented efficiently under reasonable assumptions on the likelihood function.


\section{Quantum Bayesian updating}
In order to investigate the question of whether small quantum systems can learn efficiently, we will examine the issue through the lens of Bayesian inference.
Our first objective in this section is to provide a concrete definition for what we mean by a quantum Bayesian update and show a protocol for implementing a quantum Bayesian update.  We will then show that this method cannot generically be efficient and furthermore that asymptotic improvements to the algorithm or an inexpensive error correction algorithm for its faults would violate lower bounds on Grover's search.  These results motivate our definition of ``semi--classical'' Bayesian updating in the subsequent section.

 A key assumption that we make here and in the subsequent text is that the visible and latent variables are discrete.  In other words, we assume that each experiment has a discrete set of outcomes and there are a discrete set of hypotheses that could explain the data.    Continuous problems can, however, be approximated by discrete models and we provide error bounds for doing so in~\app{discrete}.  We then invoke these assumptions in the following definition of a quantum Bayesian update.  

\begin{definition}\label{def:qbayes}
A quantum Bayesian update of a prior state $\sum_x \sqrt{P(x)} \ket{x}$ performs, for observable variable $E$ and likelihood function $P(E|x)$, the map 
$$\sum_x \sqrt{P(x)} \ket{x}\mapsto  \sum_x \sqrt{P(x|E)} \ket{x}. $$
\end{definition}

In order to formalize this notion of quantum Bayesian updating within an oracular setting we will further make a pair of assumptions.
\begin{enumerate}
\item There exists a self-inverse quantum oracle, $O_E$, that computes the likelihood function as a bit string in a quantum register: $O_E\ket{x}\ket{y} = \ket{x}\ket{y\oplus P(E|x)}$.
\item A constant $\Gamma_E$ is known such that $P(E|x) \le \Gamma_E \le 1$ for all $x$.
\end{enumerate}

An interesting consequence of the above assumptions and~\defn{qbayes} is that in general the Bayesian update is non-unitary when working on this space.  This means that we cannot  implement a quantum Bayesian update deterministically without dilating the space.
The following lemma discusses how to implement such a non--deterministic quantum Bayesian update.  It can also be thought of as a generalization of the result of~\cite{LYC14} to an oracular setting.

\begin{lemma}
Given an initial state $\sum_x \sqrt{P(x)} \ket{x}$ and $\Gamma_E : P(E|x) \le \Gamma_E$ the state $\sum_x \sqrt{P(x|E)} \ket{x}$ can be prepared using an expected number of queries to $O_E$ that is in $O(\sqrt{\Gamma_E/\langle P(x), P(E|x)\rangle})$.\label{lem:rejection}
\end{lemma}
\begin{proof}
Using a single call to $O_E$ and adding a sufficient number of ancilla qubits, we can transform the state $\sum_x \sqrt{P(x)} \ket{x}$ into
\begin{equation}
\sum_x \sqrt{P(x)} \ket{x} \ket{P(E|x)}\ket{0}.
\end{equation}
Then by applying the rotation $R_y(2 \sin^{-1}(P(E|x)/\Gamma_E))$ to the ancilla qubit, controlled on the register representing $\ket{P(E|x)}$, we can enact
\begin{equation}
\sum_x \sqrt{P(x)} \ket{x} \ket{P(E|x)}\ket{0}\mapsto \sum_x \sqrt{P(x)} \ket{x} \ket{P(E|x)}\left(\sqrt{\frac{P(E|x)}{\Gamma_E}}\ket{1} + \sqrt{1-\frac{P(E|x)}{\Gamma_E}}\ket{0} \right).\label{eq:rejection}
\end{equation}
Next the right most qubit register is measured and if a result of $1$ is obtained then the resultant state is
\begin{equation}
\frac{\sum_{x} \sqrt{P(x) P(E|x)}\ket{x}}{\sqrt{\sum_x P(x)P(E|x)}},
\end{equation}
which gives a probability distribution that corresponds to that expected by Bayes' rule.  The probability of this occurring is $\sum_x P(x) P(E|x) /\Gamma_E = \langle P(x), P(E|x) \rangle / \Gamma_E$.

Since the process is heralded, amplitude amplification can be used to boost the probability of success for the successful branch quadratically~\cite{BHM+02}.  Thus the average number of queries is in $O(\sqrt{\Gamma_E/\langle P(x), P(E|x)\rangle})$ as claimed.
\end{proof}

If the Bayesian algorithm is used solely to post--process information then one update will suffice to give the posterior distribution.  In online settings many updates will be needed to reach the final posterior distribution.  If $L$ updates are required then the probability of all such updates succeeding given a sequence of observed variables $E_1,\ldots,E_L$ is at most
\begin{equation}
P_{\rm succ} \le \sum_x P(x) \left(\max_E \frac{P(E|x)}{\Gamma_E}\right)^L\le \sqrt{\sum_x P^2(x)\sum_x \left(\max_E \frac{P(E|x)}{\Gamma_E}\right)^{2L}}\le\sqrt{\sum_x \left(\max_E \frac{P(E|x)}{\Gamma_E}\right)^{2L}}.
\end{equation}
This shows that the probability of success generically will shrink exponentially with $L$.

The exponential decay of the success probability with $L$ can be mitigated to some extent through amplitude amplification on the condition that all $L$ updates are successful  This reduces the expected number of updates needed to 
\begin{equation}
O\left(\left({\sum_x P(x) \left(\min_E \frac{P(E|x)}{\Gamma_E}\right)^L}\right)^{-1/2}\right),
\end{equation}
but this strategy is obviously insufficient to rid the method of its exponentially shrinking success probability.  Furthermore, we will see that there are fundamental limitations to our ability to avoid or correct such failures.
%


While the success probability in general falls exponentially, not all failures are catastrophic.  We show in~\app{stability} that a radius of convergence exists such that if the prior distribution is sufficiently close to a delta--function about the true value of the latent variable, and the likelihood function is well behaved, then this updating strategy will cause it to converge to the delta--function.  This convergence occurs even if the user ignores the fact that inference errors can occur and does not attempt to correct such errors.  This means that, for discrete inference problems, the computational complexity of inferring the correct latent parameter need not be infinite.

The reason why the errors in quantum rejection sampling cannot, in general, be efficiently corrected stems from the fact that quantum Bayesian inference algorithm described in~\lem{rejection} can be thought of as a generalization of Grover's algorithm~\cite{Gro96}.  Grover's problem (with one marked element) seeks to find $x = {\rm argmax}( f(x))$ where $f(x)$ is a blackbox Boolean function that is promised to have a unique $x_{\rm mark}$ such that $f(x_{\rm mark})=1$.  The generalization to multiple marked elements is similar.  The reduction between the two problems is formally proved below.

\begin{lemma}\label{lem:reduce}
Grover's problem with $N$ items and $m$ marked items reduces to Bayesian inference on a prior on $\mathbb{R}^N$.
\end{lemma}
\begin{proof}
Let $O_G$ be a Boolean function that takes the value $1$ iff $x\in X_m$ where $|X_m|=m$.  Identifying the set $X_m$ by querying this function is equivalent to Grover's problem.  Consider the following likelihood function on a two--outcome space where $1$ corresponds to finding a marked state and $0$ corresponds to finding an un--marked state:
\begin{equation}
P(1|x) = \begin{cases} 1 & {\rm if~}x\in X_m\\ 0& {\rm otherwise} \end{cases}.\label{eq:reduceLikelihood}
\end{equation}
We also have that $P(0|x)=1-P(1|x)$, but this fact is not needed for the proof.

It is then easy to see that $P(1|x)=O_G(x)$ and thus a likelihood evaluation is equivalent to a query to $O_G$.  This means that we can solve Grover's problem using the following algorithm.  \begin{enumerate}
\item Set the prior to be $P(x) = 1/N$.
\item Set $E=1$, which corresponds to pretending that an experiment was performed that found a marked entry.  
\item Compute $P(x|E) \propto P(E|x)P(x)$.  
\item Output all $x$ such that $P(x|E)=1/m$.
\end{enumerate}

The validity of this algorithm is easy to verify from~\eq{Bayes} and it is clear that the posterior distribution is a uniform distribution over $x\in X_m$.  Since $|X_m|=m$ and $P(x)$ is uniform, all elements in the support of the posterior distribution have probability $1/m$ and thus Grover's problem can be solved using Bayesian inference.  This algorithm will succeed classically using $N$ queries to $O_G$, rather than the $N-m$ queries required in the worst case scenario if a Bayesian framework is not adopted.

\end{proof}

This reduction of Grover's problem to Bayesian inference brings with it tight lower bounds on the query complexity of solving the problem~\cite{BGH+96}.  We can exploit these bounds to show limitations on quantum systems ability to perform Bayesian inference and correct erroneous measurements that occur in the application of the method of~\lem{rejection}.  The following theorem states two such restrictions, which show that the method of~\lem{rejection} cannot be trivially improved nor can its failures be inexpensively corrected. 
\begin{theorem}
Let $F$ be a blackbox quantum update algorithm that performs a quantum Bayesian update on an arbitrary pure state $\sum_{x=1}^N \alpha_x \ket{x}$ using $O(1)$ queries to the oracle $O_E$ and is heralded and has success probability $\Omega({\sum_x |\alpha_x|^2 P(E|x)/\Gamma_E})$ for $\max_x P(E|x)\le \Gamma_E\le 1$.  The following are impossible.
\begin{enumerate}
\item A blackbox algorithm $G$ capable of performing a quantum Bayesian update of an arbitrary pure quantum state of the form $\sum_{x=1}^N \alpha_x \ket{x}$ that uses $O(1)$ queries to $O_E$ that is heralded and has  success probability $\omega({\sum_x |\alpha_x|^2 P(E|x)/\Gamma_E})$.
\item A blackbox algorithm $H$ that requires $o(\sqrt{N}/\log(N))$ queries to $O_E$ to undo the effects that $F$ applies to the arbitrary state $\sum_{x=1}^N \alpha_x \ket{x}$ upon a failed update.
\item A blackbox algorithm $I$ capable of performing a quantum Bayesian update of an arbitrary pure quantum state of the form $\sum_{x=1}^N \alpha_x \ket{x}$ that, for all $O_E$, uses on average $o(\sqrt{N})$ queries to $O_E$.
\end{enumerate}
\end{theorem}
\begin{proof}
Seeking a contradiction, assume that there exists a quantum algorithm that can perform a quantum Bayesian update using $O(1)$ queries that succeeds with probability $\omega({\sum_x |\alpha_x|^2 P(E|x)})$ for any likelihood function $P(E|x)$.  Next, let us choose the likelihood function to be that used in~\eq{reduceLikelihood} in the reduction proof of~\lem{reduce} and take $\alpha_x =1/\sqrt{N}~\forall~x$.  It is clear from~\lem{rejection} that $\Gamma_E =1$ must be chosen for this problem.  Then by assumption the state $\ket{x_m}$, where $x_m \in X_m$, can be found with probability $\omega({\sum_x |\alpha_x|^2 P(E|x)})\in \omega(1/{N})$.  Since each application of the algorithm requires $O(1)$ queries and success is heralded, $o(\sqrt{N})$ queries are needed on average to learn $x_m$ using amplitude amplification~\cite{BHM+02}, which violates lower bounds for the average query complexity for Grover's problem~\cite{BGH+96}.  Therefore algorithm $G$, which is described in 1, is impossible.

Again seeking a contradiction, consider the following likelihood function with outcomes $\{1,0\}$,
\begin{equation}
P(1|x) = \begin{cases}
2/3, & x=x_m \\
1/3, & x\ne x_m
\end{cases},~P({0}|x) = 1-P(1|x)\label{eq:bayesgrover}.
\end{equation}
For each $x$, $P(1|x)$ can be computed using a single query to $O_E$ and vice versa, thus a query to this likelihood function is equivalent to a call to $O_E$.
Thus~\eq{Bayes} gives that the posterior probability after measuring $1$ is
\begin{equation}
P(x_m|1)=\frac{2P(x_m)}{1 + P(x_m)}=2P(x_m) + O(P(x_m)^2).
\end{equation}
Therefore $O(\log(N))$ measurements of $1$ suffice to amplify the probability from $1/N$ to $\Theta(1)$.  

Similarly, $P(x_m|1) \ge P(x_m)$ unless $1+P(x_m) > 2$ or $P(x_m)<0$.  Since $P(x_m)$ is a probability this is impossible.  Therefore the posterior probability is monotonically increasing with the number of successful updates.  Thus if we define $\alpha^{(1)}_x:= \alpha_x$ and $\alpha_x^{(m)}$ to be the components of the quantum state after $m+1$ quantum updates then $\sqrt{|\alpha_x^{(m)}|^2 P(1|x)}$ is a  monotonically increasing function of $m$.

In practice, it would be unlikely that $O(\log(N))$ sequential measurements would all yield $1$ (i.e. give noisy information about the marked state), but the user of a quantum Bayesian updating algorithm can always pretend that this sequence of observations was obtained (similar to~\lem{reduce}) in order to simulate the search. 
Given the observable variables follow this sequence, \lem{rejection} shows that 
there exists a quantum algorithm that can perform each such update with probability of success
\begin{equation}
\sqrt{\sum_x |\alpha_x|^2 P(1|x)} \in \Omega(1),\label{eq:Pfail}
\end{equation}
since $P(1|x) \ge 1/3$ and $\sum_x |\alpha_x|^2 =1$.

If we were not able to correct errors then~\eq{Pfail} shows that the probability of successfully inferring the marked state is $O({\rm poly}(1/N))$ since $O(\log(N))$ updates are needed; however, by assumption each failure can be corrected using $o(\sqrt{N}/\log(N))$ queries.  Therefore by attempting quantum Bayesian updates, correcting any errors that might occur and repeating until success, a successful update can be
obtained with an average number of queries that is in
\begin{equation}
o\left( \frac{\sqrt{N}}{\log(N)\sqrt{\sum_x |\alpha_x|^2 P(1|x)}}\right)\in o\left(\frac{\sqrt{N}}{\log(N)} \right),
\end{equation}
because $\sqrt{|\alpha_x^{(m)}|^2 P(1|x)}\ge \sqrt{|\alpha_x|^2 P(1|x)}$ for any $m\ge 1$.
Since $O(\log(N))$ successful quantum updates are made in the inference process, the marked state can be inferred within probability $p\in \Theta(1)$ using $o(\sqrt{N})$ queries to the likelihood function.  A
to the likelihood function is equivalent to a query to Grover's oracle and thus error correction method $H$ (described in 2) is impossible.

Finally, the impossibility of method $I$ directly follows from~\lem{reduce} and lower bounds on Grover's search.
\end{proof}

These impossibility results show that the quantum updating procedure of~\cite{LYC14} and~\lem{rejection} cannot be improved  without making assumptions about the underlying prior distributions or likelihood functions.  From this we conclude that quantum Bayesian updating, as per~\defn{qbayes}, is inneficient in general.  This means that small quantum systems that attempt to store the prior and posterior vectors as a quantum state vector cannot do so efficiently, let alone output salient properties of the state, without making such assumptions.

This inefficiency is perhaps unsurprising as exact Bayesian inference is also classically inefficient.  In particular, an efficient sampling algorithm from a distribution that is a close approximation to the  posterior distribution would imply $\P=\NP$~\cite{dagum1993approximating}.  A quantum algorithm capable of efficient Bayesian inference for general models would similarly imply that $\NP \subseteq \BQP$, which is false under reasonable complexity theoretic conjectures.

Although it may not be surprising that quantum Bayesian updating is not generically efficient, it is perhaps surprising that both it and classical updating fail to be efficient for different reasons.  Classical Bayesian updating fails to be efficient because it needs to store prior and posterior probabilities for an exponentially large number of hypotheses; however, its cost scales linearly with the number of updates used.  In contrast, quantum Bayesian updating scales polynomially with the number of hypotheses considered but scales exponentially with the number of updates.  This invites the question of whether it is possibile to combine the best features of quantum and classical Bayesian updating.  We do so in the subsequent section, wherein we show how a classical model can be stored for the system that can be reverted to in the event that a failure is observed in a quantum Bayesian update.

\section{Semi--classical Bayesian updating}
Approximations are therefore often needed to make both classical as well as quantum Bayesian inference tractable.  However, the purpose of these approximations is very different.  Classical methods struggle when dealing with probability distributions in high--dimensional spaces, and sophisticated methods like sequential Monte--Carlo approximations are often employed to reduce the effective dimension~\cite{liu_combined_2001,minka2001expectation,van2000unscented}.  However, the non--linear nature of the update rule and the problem of extracting information from the posterior distribution are not issues in the classical setting.  Our quantum algorithm has the exact opposite strengths and weaknesses: it can easily cope with exponentially large spaces but struggles emulating the non-linear nature of the update rule.

We attack the problem by making our quantum algorithm a little more classical, meaning that through out the learning process we aim to learn an approximate classical model for the posterior alongside the quantum algorithm.  This classical model allowsus to approximately re--prepare the state should an update fail throughout the updating process.  This removes the exponential scaling, but results in an approximate inference.  We refer to this procedure as \emph{quantum resampling} as it is reminiscent of resampling in sequential Monte--Carlo algorithms or other particle filter methods such as assumed-density filtering~\cite{minka2001expectation}.  In order to prepare the distribution,
we model the posterior distribution as a Gaussian distribution with mean and covariance equal to that of the true posterior.  This choice is sensible because once the Gaussian distribution is specified, the Grover--Rudolph state preparation method~\cite{GR02} can be used to prepare such states as their cumulative distribution functions can be efficiently computed.
Alternatively, for one--dimensional problems, such states could be manufactured by approximate cloning.

We are now equipped to define a semi-classical Bayesian update.
\begin{definition}
A semi--classical Bayesian update of a prior state on $\mathbb{C}^N$, for a discrete observable variable $E$, likelihood function $P(E|x)$ and family of probability distributions $F(x;\rho)$ parameterized by the vector $\rho$, maps 
$$\sum_x \sqrt{P(x)} \ket{x}\mapsto \rho: F(x;\rho) \approx P(x|E). $$\label{def:semi}
\end{definition}
We call this process semi--classical updating because it yields an approximate classical model for the posterior distribution.  This model can take many forms in principle; as an example, we could consider this model to be a Gaussian distribution that has the same mean and standard deviation as the posterior distribution.
Semi--classical Bayesian updating will be discussed in more detail in the following section, but for now we will focus on quantum Bayesian updating.


We need a means to measure the expectation values and components of the covariance matrix of the posterior for this method to work.  We provide such a method, based on the Hadamard test,  below.

\begin{lemma}
Given a unitary operator $U\in \mathbb{C}^{N\times N}$ such that $U\ket{0}=\ket{\psi}:=\sum_x \sqrt{P(x)} \ket{x}$, an observable $\Lambda=\sum_x \lambda_x \ketbra{x}{x}$ and an estimate $\lambda_0:\max_x|\lambda_x-\lambda_0|\le \Delta\lambda$, there exists a quantum algorithm to estimate $(\bra{\psi} \Lambda\ket{\psi} -\lambda_0)$ within error $\epsilon$ with probability at least $8/\pi^2$ using $ O(\Delta\lambda/\epsilon)$ applications of $U$ and queries to an oracle $O_\lambda$ such that $O_\lambda\ket{x}\ket{0}= \ket{x}\ket{\lambda_x}$.\label{lem:ampest}
\end{lemma}
\begin{proof}
By following the reasoning in~\lem{rejection}, we can prepare the following state using one query to $O_\lambda$ and one application of $U$:
\begin{equation}
\sum_{x} \sqrt{\frac{P(x)}{2}} \ket{x}\ket{\lambda_x}\left(\sqrt{1+\frac{\lambda_x-\lambda_0 }{\Delta \lambda}}\ket{1} + \sqrt{1-\frac{\lambda_x-\lambda_0 }{\Delta \lambda}}\ket{0} \right).\label{eq:hadtest}
\end{equation}
The probability of measuring $1$ is 
\begin{equation}
\sum_x \frac{P(x)}{2} \left(1+ \frac{\lambda_x - \lambda_0}{\Delta \lambda} \right)=\frac{1}{2} + \frac{\bra{\psi} \Lambda \ket{\psi}-\lambda_0}{2\Delta \lambda}.\label{eq:mark}
\end{equation}
This probability can be learned within additive error $\delta$ using $O(1/\delta^2)$ samples and hence $\bra{\psi} \Lambda \ket{\psi}-\lambda_0$ can be learned within error $\epsilon$ using $O(\Delta \lambda^2 /\epsilon^2)$ samples.

This probability can also be learned using the amplitude estimation algorithm.  Amplitude estimation requires that we mark a set of states in order to estimate the probability of measuring a state within that set.  Here we mark all states in~\eq{mark} where the rightmost qubit is $1$. The amplitude estimation algorithm then requires $O(1/\delta)$ queries to $U$ and the above state preparation method to estimate the probability to within error $\delta$ and store it in a qubit register~\cite{BHM+02}.  Amplitude estimation has a probability of success of at least $8/\pi^2$.  The result then follows from taking $\delta=\epsilon/\Delta \lambda$.
\end{proof}
%


We now turn our attention to estimating the mean and covariance of the posterior distribution that arises from quantum updating.  This is not quite a trivial application of~\lem{ampest} because our method for performing the update is non-unitary, which violates the assumptions of the Lemma.  We avoid this problem by instead estimating these moments in a two--step probability estimation process.  This approach is described in the following corollary.

\begin{corollary}
Assume that $\lambda_0$ is a vector containing each $\langle x_i\rangle$ and each $\langle x_i x_j \rangle$ evaluated over the prior and $\Delta \lambda \ge \max_k |\bra{P(x|E)} \Lambda_k \ket{P(x|E)}-\lambda_{0,k}|$ where $\Lambda_k$ is the operator corresponding to $x_i$ or $x_ix_j$  depending on the index $k$.  The mean and covariance matrix can be computed within error $\epsilon$ in the max--norm using $\tilde{O}\left(\frac{D^2\Delta \lambda \Gamma_E}{\epsilon \langle P(x), P(E|x)\rangle}\right)$ queries to $O_E$.\label{cor:method}
\end{corollary}
\begin{proof}
Our method works by classicially looping over all the components of the $\lambda_0$ vector, which we denote $\lambda_{0,k}$.
For each $k$, we then need to prepare the following state conditioned on evidence $E$ to compute the corresponding probability 
\begin{equation}
\sum_{x} \sqrt{\frac{P(x|E)}{2}} \ket{x}\ket{\lambda_x^{(k)}}\left(\sqrt{1+\frac{\lambda_x^{(k)}-\lambda_{0,k} }{\Delta \lambda}}\ket{1} + \sqrt{1-\frac{\lambda_x^{(k)}-\lambda_{0,k} }{\Delta \lambda}}\ket{0} \right),\label{eq:hadtest2}
\end{equation}
where $\lambda_x^{(k)}$ is of the form $x_i$ or $x_ix_j$ depending on the value of $k$.  

We cannot directly apply the previous lemma to learn the requisite values because the method for preparing the posterior probability distribution is non--unitary.
We address this by breaking the parameter estimation process into two steps, each of which involves learning a separate probability that we call $P(1)$ and $P(11)$.
Let $P(1)= \sum_x P(x) P(E|x)/\Gamma_E$ be the probability of performing the quantum Bayesian update.  Let $P(11)$ be the probability of both performing the update and measuring the right most qubit in~\eq{hadtest} to be $1$.  This probability is
\begin{equation}
P(11) = \frac{P(1)}{2}\left(1+\frac{\langle \Lambda_k \rangle -\lambda_{0,k}}{\Delta \lambda} \right),
\end{equation}
where $\Lambda_k$ is either of the form $x_i$ or $x_ix_j$ depending on the index $k$ and $\langle \cdot \rangle$ refers to the expectation of a quantity in the posterior state.
Therefore $\langle \Lambda_k \rangle$ can be computed from $P(11)$ and $P(1)$ via
\begin{equation}
\langle \Lambda_k \rangle = \left(2\frac{P(11)}{P(1)}-1 \right)\Delta \lambda + \lambda_{0,k},
\end{equation}
If we estimate $P(1)$ and $P(11)$ within error $O(\delta)$ then the error in $\langle \Lambda_k \rangle$ is from calculus $O(\delta/P(1))$ since $P(1)\ge P(11)$.  Therefore $\langle \Lambda_k \rangle$ can be estimated to within error $\epsilon$ if $\delta \in O(\epsilon P(1) /\Delta \lambda)$.  Bounds on the cost of amplitude estimation give the cost of this to be~\cite{BHM+02}
\begin{equation}
\tilde{O} \left(\frac{\Gamma_E \Delta \lambda}{\epsilon \sum_x P(x) P(E|x)} \right).
\end{equation}
The result then follows from noting that there are $O(D^2)$ different values of $k$ that need to be computed to learn the expectation values needed to compute the components of the posterior mean and covariance matrix.
\end{proof}
This shows that if we require modest relative error (i.e. $\Delta \lambda \in O(\epsilon)$) and $\Gamma_E$ is reasonably tight then this process is highly efficient.  In contrast, previous results that do not use these factors that incorporate apriori knowledge of the scale of these terms may not be efficient under such assumptions.

Below we combine these ideas to construct an online quantum algorithm that is capable of efficiently processing a series of $L$ pieces of data before outputting a classical model for the posterior distribution in the quantum device.  This result is key to our argument because it provides a result that one can fall back on if an update fails, thereby removing the problem of exponentially shrinking success probability at the price of only retaining incomplete information about the posterior distribution.

\begin{theorem}
Let $F(x;\mu, \Sigma)$ be a family of approximations to the posterior distribution parameterized by the posterior mean $\mu$ and the posterior covariance matrix $\Sigma$ and $\{E_k: k=1,\ldots L\}$ be a set of $L$ observable variables.  Then a semi--classical update of a quantum state $\sum_x \sqrt{P(x)} \ket{x}$ can be performed using a number of queries to $O_E$ that is in
$$
\tilde{O} \left(\frac{LD^2 \Delta  \lambda}{\epsilon \langle  P(x), \prod_{k=1}^LP(E_k|x)/\Gamma_{E_k}\rangle} \right).
$$
\end{theorem}
\begin{proof}
The algorithm is simple.  
\begin{enumerate}
\item Perform $L$ quantum Bayesian updates, but without measuring the qubits that determine whether the updates succeed or fail.
\item Use the method of~\cor{method} to learn the mean and covariance matrix of the quantum posterior distribution.
\item Return these quantities, which give a parameterization of the approximation to the posterior distribution.
\end{enumerate}

After step 1, we have from the independence of the successes that the probability of all $L$ updates succeeding is $\sum_x P(x) \prod_{k=1}^L P(E_k|x)/\Gamma_{E_k}$.  Thus step 2 can be performed using $\tilde{O}\left(\frac{D^2\Delta \lambda}{\epsilon \langle w, \prod_k P(E_k|x)/\Gamma_{E_k}\rangle}\right)$ preparations of the posterior state from~\cor{method}.  Since each such preparation requires $L$ queries to $O_E$ the total query complexity required to learn the posterior mean and variance is
\begin{equation}
\tilde{O} \left(\frac{LD^2 \Delta \lambda }{\epsilon \langle  P(x), \prod_{k=1}^LP(E_k|x)/\Gamma_{E_k}\rangle} \right)
\end{equation}

Finally, since the algorithm outputs the mean and covariance matrix of the posterior distribution, it outputs a function $F(x;\mu,\sigma)$ that captures (to within error $\epsilon$) the first two moments of the posterior distribution.  Thus the algorithm clearly performs a semi--classical update as per~\defn{semi}.
\end{proof}

The complexity of the above approximate quantum inference algorithm cannot be easily compared to that of exact Bayesian inference because both algorithms
provide very different pieces of information.  Until very recently,
no natural analogue of our quantum method could easily be found in the literature.  The result of~\cite{wiebe2015bayesian} provides such a classical analogue.  The classical query complexity of their algorithm is quadratically worse in its scaling with $\epsilon$.

Although we obtain a quadratic advantage in the scaling with $\epsilon$, it would be nice to obtain further algorithmic advantages using amplitude amplification.  Further advantages can be obtained in cases where the probabilities $P(11)$ and $P(1)$ are small by using amplitude amplification to boost these probabilities and then work backwards to infer the non--boosted probabilities.  Below we formally prove a theorem to this effect that formalizes a similar claim made informally in~\cite{WHW15}.

\begin{theorem}
Let $U$ be a unitary operator such that $U\ket{0}=\sqrt{a}\ket{\phi} + \sqrt{1-a}\ket{\phi^\perp}$ where $\braket{\phi}{\phi^\perp}=0$ for $0<a\le a_0<1$, $1-a_0\in \Theta(1)$ and let $S$ be a projector such that $S\ket{\phi}=-\ket{\phi}$ and $S\ket{\phi^\perp}=\ket{\phi^\perp}$.  Then $a$ can be estimated to within error $\epsilon$ using $\tilde{O}(\sqrt{a_0}/\epsilon)$ applications of $S$ with high probability.\label{thm:priorAmpEst}
\end{theorem}
\begin{proof}
Our proof follows the same intuition as that of the proof of amplitude estimation in~\cite{BHM+02} except rather than performing amplitude estimation on $U\ket{0}$ we use amplitude amplification to first boost the probability and then use amplitude estimation to learn the boosted probability.  The actual value of $a$ is then inferred from the amplified value of $a$ learned in the amplitude estimation step.

First by following Lemma 1 in~\cite{BHM+02} we can apply a sequence of reflection operators that contains $m$ applications of $S$ to form a unitary operation $V$ such that performs, up to a global phase,
\begin{equation}
V\ket{0} = \sin((2m+1)\sin^{-1}(\sqrt{a}))\ket{\phi} + e^{i\theta}\cos((2m+1)\sin^{-1}(\sqrt{a}))\ket{\phi^\perp}.
\end{equation}
Since $V$ is a unitary operation, amplitude estimation can be used to learn $\sin^2((2m+1)\sin^{-1}(\sqrt{a}))$ to within error $\delta$ by using Theorem 12 of~\cite{BHM+02} with probability at least $8/\pi^2$ using $O(1/\delta)$ applications of $V$.  Thus using the Chernoff bound, $\sin^2((2m+1)\sin^{-1}(\sqrt{a}))$ can be estimated within the same error tolerance using $\tilde{O}(1/\delta)$ operations with high probability.

Since $V$ contains $m$ $S$ operators, the total number of applications of $S$ needed to infer this is $O(m/\delta)$.  However, although $\sin^2((2m+1)\sin^{-1}(\sqrt{a}))$ is inferred within error $\delta$, this does not imply that $a$ is.  If we define this estimated value to be $y$ and assume that $0\le (2m+1)\sin^{-1}(\sqrt{a}) \le \pi/2$ then
\begin{equation}
a = \sin^2\left(\frac{\sin^{-1}(\sqrt{y})}{2m+1}\right).\label{eq:aeq}
\end{equation}
If there is an error of $O(\delta)$ in $y$ then Taylor analysis implies that
\begin{equation}
a=\sin^2\left(\frac{\sin^{-1}(\sqrt{y})}{2m+1}\right) + O\left(\frac{\delta}{m^2\sqrt{1-a}} \right).
\end{equation}
Since $a<1$ the error is $O(\delta/m^2)$.  Hence if we desire error $\epsilon$ in $a$ then it suffices to take $\delta\in O(\epsilon m^2)$.  Thus $\tilde O(1/\epsilon m)$ applications of $S$ are needed to infer $a$ to within error $\epsilon$.

Although this may seem to suggest that taking large $m$ always leads to a better inference of $a$, this is not necessarily true for this inversion process.  This is because if 
\begin{equation}
m > \frac{1}{2}\left(\frac{\pi}{2\sin^{-1}(\sqrt{a})}-1\right)\label{eq:mvalid}
\end{equation} then~\eq{aeq} no longer holds.  Ergo $m\in O(1/\sin^{-1}(\sqrt{a}))\in O(1/\sqrt{a})$ for small $a$.  Since the user does not know $a$, the best that can be done is to take $m\in \Theta(1/\sqrt{a_0})$ since taking $a=a_0$ also guarantees~\eq{mvalid} does not hold for $a$.  Therefore the number of applications of $G$, for $m\in \Theta(1/\sqrt{a_0})$, needed to learn $a$ within error $\epsilon$ with high probability scales as
\begin{equation}
\tilde O\left(\frac{1}{\epsilon m} \right) \in \tilde O\left(\frac{\sqrt{a_0}}{\epsilon} \right),
\end{equation}
as claimed.
\end{proof}

As an additional note, this method described in this section is not limited to tracking the values of static latent variables.  If the latent variable has itself explicit time dependence then the above approach can
be modified to robustly track its variation in time.  This is discussed in more detail in~\app{time-dep}.

\section{Adaptive experiment design}

In science and engineering, inference problems frequently involve decision variables that can
be set in order to optimize the performance of the algorithm.  For example, in the phase estimation 
algorithm such a decision variable may be the amount of time that the system is allowed to 
evolve for.  Finding locally optimal parameters for inference can be computationally challenging
(especially for online learning problems). Here, we show that our algorithm allows for quantum
computing to be used to perform Bayesian experiment design with significant
advantages over classical methods.

In practice, Bayesian experiment design is often posed in terms of finding
experiments which maximize a \emph{utility function} such as the information
gain or the reduction in a loss function. Once a utility function is
chosen, the argmax can be found by gradient ascent methods provided that the
derivatives of the utility can be efficiently computed. In particular, since
the reduction in the \emph{quadratic loss} is given by the posterior variance,
our algorithm allows for computing gradients of the corresponding utility
function.


Formally, we need to define two quantities: the loss function and the Bayes
risk. In doing so, we will assume without loss of generality that the model
parameters are renormalized such that all components of $x$ lie in $[0, 1]$.
The loss function represents a penalty assigned to errors in the in our
estimates of $x$. We consider here the multiparameter generalization of the
mean-squared error, the quadratic loss. For an estimate $\hat{x}$,
\begin{equation}
    \mathcal{L}(x, \hat{x}) = \left(x - \hat{x}\right)^\TT \left(x - \hat{x}\right).
\end{equation}
Letting $\hat{x}$ be the Bayesian mean estimator for the posterior $P(x | d, c)$
and considering the single-parameter case,
\begin{equation}
    \mathcal{L}(x, P(x|d,c)) = \left(x-\int P(x'|d,c) x' \mathrm{d}x'\right)^2.
\end{equation}
Having defined the loss function, the risk is the expectation of the loss over
experimental data, $\expect{\mathcal{L}(x, \hat{x})}{d}$, where $\hat{x}$ is
taken to depend on the experimental data. The Bayes risk is then the
expectation of risk over both the prior distribution and the outcomes,
\begin{gather}
    \begin{aligned}
        \mathcal{R}(x, P(x)) & = \expect{\mathcal{L}(x, P(x | d, c))}{d, x \sim P(x)} \\
                             & = \int P(x)\int P(d|x,c) \left(x-\int P(x'|d,c) x' \mathrm{d}x'\right)^2 \mathrm{d}x\mathrm{d}d.
    \end{aligned}
\end{gather}
The Bayes risk for the quadratic loss function is thus the trace of the posterior covariance matrix,
averaged over possible experimental outcomes.
We want to find $c$ that minimizes the Bayes risk, so that a reasonable utility function to optimize for is the negative posterior variance,
\begin{equation}
    \label{eq:utility}
    \mathcal{U}(P(x),c) = -\int P(x)\int P(d|x,c) \left(x-\int P(x'|d,c) x' \mathrm{d}x'\right)^2 \mathrm{d}x\mathrm{d}d.
\end{equation}

The application of our algorithm is now made clear: like classical particle
filtering methods, our algorithm can estimate expectation values over
posterior distributions efficiently. Thus, $\mathcal{U}$ can be calculated using quantum resources,
including in cases where classical methods alone fail.
In the finite dimensional setting that we're interested in we simply replace these integrals by sums over the corresponding variables.
The derivatives of $\mathcal{U}$ can then be approximated for small but finite $h$ as 
\begin{equation}
\frac{\partial \mathcal{U}(P(x),c)}{\partial c_j} = \frac{\mathcal{U}(P(x),c+h\hat{c_j})-\mathcal{U}(P(x),c)}{h} + O(h^2).
\end{equation}
Thus if $c$ consists of $C$ different components then $O(C)$ calculations of the utility function are needed to estimate the gradient for a finite value of $h$.
This is the intuition behind our method, the performance of which is given in the following theorem.
\begin{theorem}
Assume that the prior distribution $P(x)$ has support only on the interval $x\in [0,1]$ and that the observable varible $E$ has support only on $D$ distinct values; then each component of the gradient of $\mathcal{U}$ can be computed within error $\epsilon$ using on average $\tilde{O}\left(D\sqrt{\max_{c,j}\left|\frac{\partial^3 U(c)}{\partial c_j^3} \right|}/\epsilon^{3/2}\right)$ queries to the likelihood function and the prior, for $\epsilon \le \min_d \int P(d|x)P(x) \mathrm{d}x/2$.
\end{theorem}
\begin{proof}
The utility function can be directly computed on a quantum computer, but doing
so is challenging because of the need to coherently store the posterior means
of the distribution.  We simplify this by expanding the square in~\eq{utility}
to find
\begin{align}
\mathcal{U}(P(x),c) = &-\iint P(x) P(d|x,c) x^2 \mathrm{d}x\mathrm{d}d\nonumber\\
&+2\iiint P(x) P(d|x,c) P(x'|d,c) xx' \mathrm{d}x' \mathrm{d}d\mathrm{d}x\nonumber\\
&-\iiiint P(x) P(d|x,c) P(x'|d,c)P(x''|d,c) x'x'' \mathrm{d}x''\mathrm{d}x' \mathrm{d}d\mathrm{d}x.\label{eq:multiutility}
\end{align}
We then compute each of these terms individually and combine the results classically to obtain an estimate of $\mathcal{U}$.

The double integral term in~\eq{multiutility} is the easiest to compute.  It can be computed by preparing the state
\begin{equation}
\sum_x \sqrt{P(x)} \ket{x} \frac{1}{\sqrt{D}}\sum_{d=1}^D \ket{d}\left(\sqrt{{P(D|x,c)x^2}}\ket{1} + \sqrt{1 -P(D|x,c)x^2}\ket{0} \right). 
\end{equation}
The probability of measuring the right most qubit to be $1$ is $$\sum_x\sum_d P(x)P(D|x,c)x^2/D \le 1/D.$$  Therefore the desired probability can be found by estimating the likelihood of observing $1$ divided by the total number of outcomes $D$.  A direct application of amplitude estimation gives that the expectation value can be learned within error $\epsilon$ using $\tilde{O}(D/\epsilon)$ preparations of the initial state and evaluations of the likelihood function.  

Since the probability of success is known to be bounded above by $1/D$, \thm{priorAmpEst} implies that $\tilde{O}(\sqrt{D}/\epsilon_0)$ state preparations are needed to estimate the integral if we define $S$ to be a reflection operator that imparts a phase if and only if the ancilla qubit equals $1$.

The numerator can be estimated in exactly the same fashion, by preparing the state
\begin{equation}
\sum_x \sqrt{P(x)} \ket{x}\sum_{x'} \sqrt{P(x')} \ket{x'} \left(\sqrt{P(d|x',c){P(d|x,c)}xx'}\ket{1} + \sqrt{1 -P(d|x',c)P(d|x,c)xx'}\ket{0} \right),
\end{equation}
Note that the numerator, $N(d|c)$, is not $\Theta(1)$: it is in fact $O(P^2(d))$ as seen by the Cauchy--Schwarz inequality and $x\in [0,1]$
\begin{equation}
\sum_x \sum_{x'} P(x) P(x') P(d|x,c) P(d|x',c) x x' \le \left(\sum_x P(x) P(d|x,c)\right)^2=P_d^2.
\end{equation}

The triple integral in~\eq{multiutility} is much more challenging.  It can be expressed as
\begin{equation}
\iiint P(x) P(d|x,c) P(x'|d,c) xx' \mathrm{d}x' \mathrm{d}d\mathrm{d}x\nonumber = \iiint P(x) P(d|x,c) \frac{P(d|x',c)P(x')}{\int P(d|x',c) P(x')\mathrm{d}x'} xx' \mathrm{d}x' \mathrm{d}d\mathrm{d}x.
\end{equation}
The integral over $d$ in this expression is difficult to compute in superposition.  So instead, we forgo directly integrating over $d$ using the quantum computer and instead compute the integrand quantumly and classically integrate over $d$.  In many models $D$ will be small ($D=2$ is not uncommon) hence a polynomial reduction in the scaling with $D$ will often not warrant the additional costs of amplitude amplification.

For fixed $d$, the first step is to compute $P(d):=\int P(d|x,c) P(x)\mathrm{d}x$, which can be estimated by preparing the state
\begin{equation}
\sum_x \sqrt{P(x)} \ket{x} \left(\sqrt{{P(d|x,c)}}\ket{1} + \sqrt{1 -P(d|x,c)}\ket{0} \right),\label{eq:update1}
\end{equation}
and estimating, $P(d)$, the probability that the right--most qubit is $1$, which is the required probability.   This can be learned within error $\epsilon$ using amplitude estimation, which requires $\tilde{O}(1/\epsilon)$ queries to the initial state and the likelihood oracle~\cite{BHT+00}.

For simplicity let us define the integral to be $N(d|c)/P(d)$ and the approximation to the integral as $\tilde{N}(D|c)/\tilde{P}(d)$.  We then see from the triangle inequality that if we estimate the denominator to within error $\epsilon_0\le P(d)/2$ then
\begin{align}
\left|\frac{\tilde{N}(d|c)}{\tilde{P}(d)} -\frac{{N}(d|c)}{{P}(d)}\right|&\le \left|\frac{\tilde{N}(d|c)}{\tilde{P}(d)} -\frac{{N}(d|c)}{\tilde {P}(d)}\right|+\left|\frac{{N}(d|c)}{\tilde{P}(d)} -\frac{{N}(d|c)}{{P}(d)}\right|.\nonumber\\
&\le \frac{1}{P(d)-\epsilon_0}\left|{\tilde{N}(d|c)} -{{N}(d|c)}{}\right|+P(d)^2\left|\frac{1}{\tilde{P}(d)} -\frac{1}{{P}(d)}\right|\nonumber\\
&\le \frac{2}{P(d)}\left|{\tilde{N}(d|c)} -{{N}(d|c)}{}\right|+P(d)\left|\frac{1}{1-\epsilon_0/P(d)} -1\right|\nonumber\\
&\le \frac{2}{P(d)}\left|{\tilde{N}(d|c)} -{{N}(d|c)}{}\right|+\epsilon_0.\label{eq:2IntegralTerms}
\end{align}
Therefore under these assumptions it is necessary to estimate $N(d|c)$ to within error $O(\epsilon_0/P(d))$ to achieve error $O(\epsilon_0)$.
We can accelerate this inference process by observing that 
\begin{equation}
N(d|c)\le (P(d)+\epsilon_0)^2\in O(P^2(d)),
\end{equation}
since $\epsilon_0 \le P(d)/2$.  \thm{priorAmpEst} can then be used to estimate $N(d|c)$ within error $\delta$ using $\tilde O(P(d)/\delta)$ queries.  Since we need error $\epsilon_0P(d)$ the number of query operations needed to infer $N(d|c)$ within error $\epsilon_0$ is in $\tilde O(1/\epsilon_0)$.  This process needs to be repeated classically $D$ times so the total cost is $\tilde O(D/\epsilon_0)$ for this step as well.  Thus we see from~\eq{2IntegralTerms} that the total error can be made less than $\epsilon_0$, with high probability, using a number of queries that scales as $\tilde{O}(D/\epsilon)$.


The analysis of the quadrouple integral is exactly the same and requires $\tilde O(D/\epsilon)$ queries on average.  Thus the cost of evaluating the utility function to within error $\epsilon$ with high probability is $\tilde O(D/\epsilon)$.  

Given an algorithm that can compute $U(c)$ using a number of queries that scales as $\tilde{O}(D/\epsilon)$, we can estimate the derivative using a centered difference formula.  In particular we know that

\begin{equation}
\left|\frac{\partial U(c)}{\partial c_j}-\frac{U(c+\delta_j)-U(c-\delta_j)}{2\delta} \right|\le \max_{c,j} \left|\frac{\partial^3 U(c)}{\partial c_j^3} \right|\frac{\delta^2}{6},
\end{equation}
where $|\delta_j|=\delta$ and $\delta_j$ is a vector parallel to the unit vector $c_j$.
Since we cannot compute $U(c\pm\delta_j)$ exactly, the error we want to bound is
\begin{align}
&\left|\frac{\partial U(c)}{\partial c_j}-\frac{\tilde U(c+\delta_j)-\tilde U(c-\delta_j)}{2\delta} \right|\nonumber \\
&\qquad \le \left|\frac{\partial U(c)}{\partial c_j}-\frac{ U(c+\delta_j)- U(c-\delta_j)}{2\delta} \right|+\left|\frac{ U(c+\delta_j)- U(c-\delta_j)}{2\delta}-\frac{ \tilde U(c+\delta_j)- \tilde U(c-\delta_j)}{2\delta} \right|,
\end{align}
where $\tilde U$ is the approximation to the utility function that has error at most $\epsilon_0$.  The error is then
\begin{equation}
O\left(\max_{c,j} \left|\frac{\partial^3 U(c)}{\partial c_j^3} \right|{\delta^2}+\frac{\epsilon_0}{\delta}\right).
\end{equation}
Since $\delta$ is a free parameter that we will choose to make both sources of error equivalent.  This corresponds to $\delta= \epsilon_0^{1/3}/\max_{c,j} \left|\frac{\partial^3 U(c)}{\partial c_j^3} \right|^{1/3}$.  This gives an overall error of
\begin{equation}
O\left({\epsilon_0^{2/3}}{\max_{c,j}\left|\frac{\partial^3 U(c)}{\partial c_j^3} \right|^{1/3}} \right).
\end{equation}
If we wish to make this error $\epsilon$ then it suffices to take $\epsilon_0 = \epsilon^{3/2}/\sqrt{\max_{c,j}\left|\frac{\partial^3 U(c)}{\partial c_j^3} \right|}$.  Since the cost of computing $U(c\pm \delta_j)$ within error $\epsilon_0$ with high probability is $\tilde O(D/\epsilon_0)$ the cost estimates follow.
\end{proof}

This shows that we can use quantum techniques to achieve a polynomial speedup over classical methods for computing the gradient using sampling, which would require $O(D/\epsilon_0^2)$ queries.  It is also worth noting that high--order methods for estimating the gradient may be useful for further improving the error scaling.

Another interesting feature of this approach is that we do not explicitly use the qubit string representation for the likelihood to prepare states such as~\eq{update1}.   Similar states could therefore also be prepared for problems such as quantum Hamiltonian learning~\cite{WGF+14} by eschewing a digital oracle and instead using a quantum simulation circuit that marks parts of the quantum state that correspond to measurement outcome $d$ being observed.  This means that these algorithms can be used in concert with quantum Hamiltonian learning ideas to efficiently optimize experimental design, whereas no efficient classical method exists to do so because of the expense of simulation.

\section{Quantum Bayesian updating using repetition codes}
While we have addressed a semi--classical form of learning for quantum Bayesian inference, an interesting remaining question is whether the form of quantum
Bayesian inference that we consider has a well defined classical limit.  The hope would be that such a protocol would also be an approximate Bayesian method, but would not be susceptible to the probabilistic failures that plague the quantum approach.  We can reach such a limit by using a repetition code to perform a protocol that is similar to semi--classical updating, but does not involve storing classical information.  We also focus on the one-dimensional case in the following, but generalization to the multi-dimensional case is straightforward.

The repetition code that allows us to reach the classical limit of the quantum algorithm is trivial:
\begin{equation}
\sum_{j} \sqrt{P(x_j)} \ket{x_j} \mapsto \ket{\psi}:= \left(\sum_j \sqrt{P(x_j)} \ket{x_j}\right)^{\otimes K}.
\end{equation}
In order to learn the mean from such a state without destroying it, we need to add an additional register that stores an estimate of the mean-value to a fixed number of bits of precision.  This can be achieved using a simple arithmetic circuit.  We denote this state as
\begin{equation}
\sum_{x_1,\ldots,x_K} \sqrt{P(x_1)\cdots P(x_K)} \ket{x_1\ldots x_K} \ket{\bar{x}(x_1\ldots x_K)},
\end{equation}
where $\bar{x}$ is an approximation to the mean that is truncated to give error $\Delta\le \mu$.  For simplicity, we drop the explicit dependence of $\bar{x}$ on $x$ in the following.

Let $\mu= \sum_j P(x_j) x_j$ be the true mean.  Then as each of the distributions over the constituent $x_j$ is independent and assuming that $x_j \le X_{\rm max}$, the Chernoff bound states that
\begin{equation}
P\left(\left|\bar{x} - \mu\right|\ge \Delta \right)\le e^{-\frac{\Delta^2 K}{3\mu X_{\rm max}}}.
\end{equation}
Thus the probability of measuring a mean that deviates more than $\Delta$ from $\mu$ is at most $\delta$ if
\begin{equation}
K\ge \frac{3\mu X_{\rm max}}{\Delta^2}\ln\left(\frac{1}{\delta} \right).
\end{equation}
This implies that for every $\delta>0$ and every discretization error $\Delta$ there exists a value of $K$ such that the probability of measuring the discretized mean to be $\mu$ is at least $1-\delta$.

Let $\ket{\phi} = (\openone \otimes \ketbra{\mu}{\mu}) \ket{\psi} / |(\openone \otimes \ketbra{\mu}{\mu}) \ket{\psi}|$ then
\begin{equation}
|\braket{\psi}{\phi}|^2 = \frac{|\bra{\psi} (\openone \otimes \ketbra{\mu}{\mu}) \ket{\psi}|^2}{|(\openone \otimes \ketbra{\mu}{\mu}) \ket{\psi}|^2} \ge \frac{1-\delta}{|(\openone \otimes \ketbra{\mu}{\mu}) \ket{\psi}|^2} \ge 1-\delta.
\end{equation}
Thus up to error $O(\delta)$, we can treat the state after learning the mean as identical to the state that existed before learning $\mu$.  Ergo despite the fact that the $x_i$ used in the distribution are no longer identically distributed, we can treat them as if they were while incurring an error of at most $\delta$ in the estimate of $P(|\bar{x} - \mu|\le \Delta)$.  
From the triangle inequality, it is then straight forward to see that after $L$ such steps that the total error incurred in the final state (as measured by the trace distance) is at most $L\sqrt{\epsilon}$, which can be made at most $\sqrt{\epsilon}$ by choosing 
\begin{equation}
K\ge \frac{3\mu X_{\rm max}}{\Delta^2}\ln\left(\frac{L^2}{{\epsilon}} \right).
\end{equation}
This in turn means that the error in the inference of $\mu$ after $L$ steps is at most $X_{\rm max} \epsilon$.  The exact same argument can be applied to learn the mean--square value of $x$ and so the standard deviation can be learned in a similar fashion.

If $K$ is sufficiently large, then any branches that fail can be immediately repopulated by a distribution from a two--parameter family of distributions $F(x;\mu,\sigma^2)$.  This further carries an advantage because it does not necessitate that the entire distribution be approximated at each iteration, unlike semi--classical updating.

This shows that a redundant encoding can be used in order to protect the low--order moments against the effects of measurement.  Therefore means that even if some updates fail then these results can be erased and replaced with a Gaussian approximation to the posterior distribution (for example).  An explicit classical register is not needed in this approach, although since the entanglement of the expectation value register with the remaining qubits approaches zero.  In this sense, it becomes a classical register and this result can also be thought of as an examination of the classical limit of quantum Bayesian updating.

While this shows that repetition codes can allow the algorithm to proceed without classical memory, it makes substantial demands on the memory.  For even modest problems, it is likely to require thousands of copies of the state in order to be able to resist the effects of measurement back action on the state.  This shows that while error correction can allow quantum systems able to learn efficiently without classical memory, the resulting systems will seldom be small.  This suggests that there may be a tradeoff between system size and robustness that may make learning in small quantum systems highly challenging.

\section{Conclusion}
Our main contribution of this paper can be thought of as an analysis of the ability of small quantum systems' capacity to learn.  In it we have examined a class of quantum learning algorithms that require only logarithmic memory to update a register that stores its beliefs about a latent variable $x$ that it must infer from a set of observable variables $E$ that whose likelihoods are only known through access to a quantum oracle.  We show that such algorithms cannot in general be efficient, but provide a semi--classical algorithm that uses classical memory to circumvent this problem within the context of approximate Bayesian inference.

Our semi--classical algorithm has a number of performance advantages over classical methods.  It also can leverage quantum superposition to provide quadratic advantages for experiment design and also can be used to track the motion of time--dependent latent variables.  We further demonstrate the need to store the model in classical memory is in principle superflous because a quantum repetition code can be used to robustly encode this information in the quantum state.  This suggests that while small quantum systems may not be able to efficiently learn (especially in an online setting), redundant information can be used to protect the knowledge gained by the quantum system against the potentially destructive impacts of the non--linear transformations required by Bayesian inference.  This supports the conjecture that error correction is intimately linked to  learning.  

Although our work suggests that small quantum systems may face substantial difficulties when trying to perform Bayesian inference in an oracular setting, much more work is needed in order to provide a complete answer to the question.  In particular, a firm definition is needed in order to address the question of what learning even means for quantum systems outside of the confines of the definition we implicitly assume through our consideration of Bayesian updating.  A suitable answer to this question may not only shed light on the nature of learning in physical systems but also help us come to grips with the limitations that arises from trying to reason using agents whose ephemeral memories are kept in quantum states of matter.

\begin{acknowledgments}
We would like to thank J. Combes and J. Yard for valuable feedback and discussion as well as J. Emerson for suggesting the idea of quantum learning agents.
\end{acknowledgments}

\bibliography{qsmc2}


\appendix
\section{Asymptotic Stability of Updating}\label{app:stability}

Interestingly, this process of classically learning a model for the posterior need not be repeated forever.  If the true model has sufficient support in the final posterior then classical feedback is irrelevant because the quantum algorithm will converge to the true model as $L\rightarrow \infty$ if $P(E|x) \ne P(E|y)$ for all $x\ne y$, regardless whether success or failure is observed.  This is summarized in the following theorem.
\begin{theorem}
There exists $\delta>0$ such that if $|\ket{\psi} - \ket{x}| \le \delta$ then the method of~\lem{rejection} converges to $\ket{x}$ if the failure and success branches are treated equivalently and $P(E|x) \ne P(E|y)$ for all $x\ne y$.
\end{theorem}
\begin{proof}
The algorithm that results from ignoring whether success or failure is measured in the method of~\lem{rejection} can be studied by examining the map that results from tracing over the success or failure register.
First, let us assume that the likelihood function is non--degenerate, meaning that $P(E|x)$ is unique for all $x$.  Then applying~\eq{rejection} we see that
\begin{equation}
\ket{x} \ket{P(E|x)}\ket{0}\mapsto \sqrt{P(x)} \ket{x} \ket{P(E|x)}\left(\sqrt{\frac{P(E|x)}{\Gamma_E}}\ket{1} + \sqrt{1-\frac{P(E|x)}{\Gamma_E}}\ket{0} \right).
\end{equation}
Because the state is not entangled, measuring the right most qubit does not affect the remaining state.  Therefore the transformation given by~\lem{rejection} has each computational basis state as an eigenvector with eigenvalue $1$.

Now let us assume that we apply the algorithm to the state $\ket{\psi} = \ket{x} +\delta \ket{x_2} + O(\delta^2)$ for $\delta\ll 1$.  It then follows from tracing over the register that the resultant state is
\begin{equation}
\ketbra{x}{x} + \delta\left(\sqrt{\frac{P(E|x)P(E|x_2)}{\Gamma_E^2}} + \sqrt{\left(1-\frac{P(E|x)}{\Gamma_E}\right)\left(1-\frac{P(E|x_2)}{\Gamma_E}\right)} \right)(\ketbra{x}{x_2}+\ketbra{x_2}{x}) + O(\delta^2).
\end{equation}
It is straightforward to see from calculus that the $O(\delta)$ term is maximized when $P(E|x)=P(E|x_2)$, which is forbidden under our assumptions.  Furthermore, the coefficient is at most $1$, ergo the resultant state can be expressed as
\begin{equation}
(\ket{x} + c\delta \ket{x_2})(\bra{x} + c\delta \bra{x_2}) + O(\delta^2),
\end{equation}
for $0\le c<1$.  Therefore the resulting state is equivalent to the initial state, but with the component orthogonal to it reduced by a factor of $c$.  This means that the algorithm converges to $\ket{x}$ after a sufficient number of repetitions given that $\delta\ll 1$.

Now let us imagine that an initial state of the form $\ket{x} + \delta \sum_{y\ne x} a_y \ket{y} + O(\delta^2)$ is prepared.  The density operator that results from applying the mapping in~\eq{rejection} and tracing over the last qubit is
\begin{equation}
\ketbra{x}{x} + \delta\sum_{y\ne x}a_y\left(\sqrt{\frac{P(E|x)P(E|y)}{\Gamma_E^2}} + \sqrt{\left(1-\frac{P(E|x)}{\Gamma_E}\right)\left(1-\frac{P(E|y)}{\Gamma_E}\right)} \right)(\ketbra{x}{y}+\ketbra{y}{x}) + O(\delta^2).
\end{equation}
Following the same argument it is clear that there exist $0\le c_y<1$ such that the resultant state is
\begin{equation}
\left(\ket{x} + \sum_{y\ne x}a_y c_y\delta \ket{y}\right)\left(\bra{x} + \sum_{y\ne x}c_y\delta \bra{y}\right) + O(\delta^2),
\end{equation}
It is clear that the resultant state can be written in the form is $\ket{x} + c\delta \ket{\psi} + O(\delta^2)$ where $0\le c\le \max_y c_y < 1$.  This shows that the algorithm converges to $\ket{x}$ even if the initial perturbation is a superposition of basis states.
\end{proof}

\section{Discretization Errors}\label{app:discrete}
Apart from the quantum resampling step, the only source of error that emerges in our inference algorithm is from the discretization of the problem.  We assume here that the underlying probability distribution $P(x)$ and the likelihood function $P(E|x)$ are differentiable functions of $x$ and assume without loss generality that $x\in [0,1]^D$.  We furthermore assume that the mesh used to approximate the probability distribution is uniform and a gridspacing of $\Delta x$ is used in each direction.  This means that the number of points is 
\begin{equation}
N= (\Delta x)^{-D}.
\end{equation}

For notational simplicity,  we take
\begin{equation}
\langle P(x), P(E|x) \rangle := \int P(x) P(E|x) \mathrm{d}^D x.
\end{equation}
We then give our main theorem below using this notation.

\begin{theorem}
Let $P(E|x)$ be a differentiable function of $x\in [0,1]^D$ such that $0< \max_E|\nabla P(E|x) |_{\rm max} \le \Lambda$ and assume $\langle P(E|x), P(x)\rangle \ne 0$.   A component of the posterior mean, $[x]_k$, can then  be approximated for any $k\in \{1,\ldots, K\}$ within error $\epsilon$ by simulating a Bayes update of $P(x)$ on a uniform mesh of $[0,1]^D$ with mesh spacing $\Delta x$ where 
$$\Delta x \le \min_E \frac{\epsilon \langle P(E|x),P(x)\rangle^2}{\langle P(E|x),P(x)\rangle^2+3{D} \Lambda},$$
and
$$\epsilon \le \frac{\langle P(E|x),P(x)\rangle^2+ 3D\Lambda}{2 D\Lambda\langle P(E|x),P(x)\rangle}.$$
\label{thm:errbd}
\end{theorem}
\begin{proof}
We employ the following approximation scheme.  Let $V_j$ be a hypercube of volume $\Delta x^D$ with centroid $\bar{x}_j$.  We then approximate the prior distribution within the hypercube as $P(x) \approx  \delta (x-\bar{x}_j)\int_{V_j} P(x) \mathrm{d}^D x$.  Our goal is to bound the error that this approximation incurs in the posterior mean.

We first analyze the error in approximating the probability assigned to each hypercube $V_j$ after a Bayesian update
\begin{align}
&\left|\frac{\int_{V_j}P(E|x) P(x) \mathrm{d}^D x}{\int P(E|x) P(x) \mathrm{d}^D x} - \frac{P(E|\bar{x}_j)\int_{V_j} P(x) \mathrm{d}^D x}{\sum_jP(E|\bar{x}_j)\int_{V_j} P(x) \mathrm{d}^D x} \right|\nonumber\\
&\qquad\le \left|\frac{\int_{V_j}P(E|x) P(x) \mathrm{d}^D x}{\int P(E|x) P(x) \mathrm{d}^D x} - \frac{P(E|\bar{x}_j)\int_{V_j} P(x) \mathrm{d}^D x}{\int P(E|x) P(x) \mathrm{d}^D x} \right|+\left|\frac{P(E|\bar{x}_j)\int_{V_j} P(x) \mathrm{d}^D x}{\int P(E|x) P(x) \mathrm{d}^D x} - \frac{P(E|\bar{x}_j)\int_{V_j} P(x) \mathrm{d}^D x}{\sum_jP(E|\bar{x}_j)\int_{V_j} P(x) \mathrm{d}^D x} \right|.\label{eq:triangleeq}
\end{align}

From Taylor's remainder theorem and the triangle inequality, we then see that
\begin{align}
\int_{V_j} (P(E|x)-P(E|\bar{x}_j))P(x) \mathrm{d}^Dx\le D\Lambda \Delta x \int_{V_j} P(x)\mathrm{d}^Dx,
\end{align}
which implies that
\begin{equation}
 \left|\frac{\int_{V_j}P(E|x) P(x) \mathrm{d}^D x}{\int P(E|x) P(x) \mathrm{d}^D x} - \frac{P(E|\bar{x}_j)\int_{V_j} P(x) \mathrm{d}^D x}{\int P(E|x) P(x) \mathrm{d}^D x} \right| \le \frac{D\Lambda \Delta x\int_{V_j} P(x)\mathrm{d}^Dx}{\int P(E|x) P(x) \mathrm{d}^D x}.\label{eq:term1}
\end{equation}

Now looking at the remaining term in~\eq{triangleeq} we see that setting $$\sum_j \int_{V_j} \Delta P(E|x) P(x) \mathrm{d}^Dx := \sum_j P(E|\bar{x}_j) \int_{V_j} P(x) \mathrm{d}^Dx -\int P(E|x) P(x) \mathrm{d}x ,$$
Using this definition, we can upper bound 
\begin{align}
&\left|\frac{1}{\int P(E|x) P(x) \mathrm{d}^D x}- \frac{1}{\sum_j P(E|\bar{x}_j)\int_{V_j} P(x)\mathrm{d}^Dx}\right|\nonumber\\
&\qquad\le\frac{1}{\int P(E|x) P(x) \mathrm{d}^D x}\max \left| 1 - \frac{1}{1-\sum_j \Delta P(E|\bar{x}_j)\int_{V_{j}} P(x) \mathrm{d}^D x\big/\int P(E|x) P(x) \mathrm{d}^D x}\right|
\end{align}
For the moment, let us assume that $\Delta x$ is chosen such that
\begin{equation}
|\Delta P(E|x)|\le {D}\Lambda \Delta x \le \int P(E|x) P(x) \mathrm{d}^Dx/2.\label{eq:assumption}
\end{equation} 
We will see that this is a consequence of the bound on $\epsilon$ in the theorem statement.  Then, using the fact that for all $|z|\le 1/2$, $|1/(1+z) -1|\le 2|z|$
\begin{equation}
\frac{1}{\int P(E|x) P(x) \mathrm{d}^D x}\max \left| 1 - \frac{1}{1-\sum_j \int_{V_{j}}\Delta P(E|\bar{x}_j) P(x) \mathrm{d}^D x\big/\int P(E|x) P(x) \mathrm{d}^D x}\right|\le  \frac{2D\Lambda \Delta x}{\left(\int P(E|x) P(x) \mathrm{d}^D x\right)^2}.\label{eq:term2}
\end{equation}
This implies that
\begin{equation}
\left|\frac{P(E|\bar{x}_j)\int_{V_j} P(x) \mathrm{d}^D x}{\int P(E|x) P(x) \mathrm{d}^D x} - \frac{P(E|\bar{x}_j)\int_{V_j} P(x) \mathrm{d}^D x}{\sum_jP(E|\bar{x}_j)\int_{V_j} P(x) \mathrm{d}^D x} \right|\le  \frac{2D\Lambda \Delta xP(E|\bar{x}_j)\int_{V_j} P(x) \mathrm{d}^D x}{\left(\int P(E|x) P(x) \mathrm{d}^D x\right)^2}
\end{equation}
Thus from~\eq{term1},~\eq{term2} and~\eq{triangleeq}
\begin{equation}
\left|\frac{\int_{V_j}P(E|x) P(x) \mathrm{d}^D x}{\int P(E|x) P(x) \mathrm{d}^D x} - \frac{P(E|\bar{x}_j)\int_{V_j} P(x) \mathrm{d}^D x}{\sum_jP(E|\bar{x}_j)\int_{V_j} P(x) \mathrm{d}^D x} \right|
\le \frac{3D\Lambda\Delta x\int_{V_j} P(x) \mathrm{d}^D x}{\left(\int P(E|x) P(x) \mathrm{d}^D x\right)^2}.\label{eq:proberror}
\end{equation}

Now let $[x]_k$ be the $k$--th component of the vector $x$.  It then follows that the posterior mean of that component of the model vector obeys
\begin{align}
\left|\int P(x|E) x_k \mathrm{d}^D x- \sum_j  \frac{[\bar{x}_j]_k P(E|\bar{x}_j)\int_{V_j} P(x) \mathrm{d}^D x}{\sum_jP(E|\bar{x}_j)\int_{V_j} P(x) \mathrm{d}^D x}\right| &\le \left|\int P(x|E) x_k \mathrm{d}^D -\sum_j \int_{V_j} P(x|E)[\bar{x}_j]_k \mathrm{d}^D x \right|\nonumber\\
& + \left|\sum_j \int_{V_j} P(x|E)[\bar{x}_j]_k \mathrm{d}^D x- \sum_j  \frac{[\bar{x}_j]_k P(E|\bar{x}_j)\int_{V_j} P(x) \mathrm{d}^D x}{\sum_jP(E|\bar{x}_j)\int_{V_j} P(x) \mathrm{d}^D x} \right|.\label{eq:meantriangle}
\end{align}
Since $0\le [\bar{x}_j]_k \le 1$ and the sum of the prior probability is $1$,~\eq{proberror} implies
\begin{equation}
\left|\sum_j \int_{V_j} P(x|E)[\bar{x}_j]_k \mathrm{d}^D x- \sum_j  \frac{[\bar{x}_j]_k P(E|\bar{x}_j)\int_{V_j} P(x) \mathrm{d}^D x}{\sum_jP(E|\bar{x}_j)\int_{V_j} P(x) \mathrm{d}^D x} \right| \le  \frac{3D\Lambda\Delta x}{\left(\int P(E|x) P(x) \mathrm{d}^D x\right)^2}.\label{eq:meantriangle1}
\end{equation}
Similarly,
\begin{equation}
\left|\int P(x|E) x_k \mathrm{d}^D -\sum_j \int_{V_j} P(x|E)[\bar{x}_j]_k \mathrm{d}^D x \right|\le \sum_j \int_{V_j} P(x|E) \mathrm{d}^D x \Delta x = \Delta x.\label{eq:meantriangle2}
\end{equation}
Therefore \eq{meantriangle}, \eq{meantriangle1} and \eq{meantriangle2} imply that the error in the posterior mean is
\begin{equation}
\left|\int P(x|E) x_k \mathrm{d}^D x- \sum_j  \frac{[\bar{x}_j]_k P(E|\bar{x}_j)\int_{V_j} P(x) \mathrm{d}^D x}{\sum_jP(E|\bar{x}_j)\int_{V_j} P(x) \mathrm{d}^D x}\right| \le \Delta x\left(1+\frac{3D\Lambda}{\left(\int P(E|x) P(x) \mathrm{d}^D x\right)^2} \right).\label{eq:meanbound}
\end{equation}
Simple algebra then shows that the error in the approximate posterior mean is at most $\epsilon$ if
\begin{equation}
\Delta x \le \max_E \frac{\epsilon \langle P(E|x),P(x)\rangle^2}{\langle P(E|x),P(x)\rangle^2+3{D} \Lambda}.\label{eq:Deltaxbd}
\end{equation}

Eq.~\eq{assumption} is a key assumption behind~\eq{Deltaxbd}.  It is then easy to see from algebra that the assumption is implied by~\eq{Deltaxbd} if
\begin{equation}
\epsilon \le \frac{\langle P(E|x),P(x)\rangle^2+ 3D\Lambda}{2 D\Lambda\langle P(E|x),P(x)\rangle},
\end{equation}
as claimed.
\end{proof}

This theorem shows that if the derivatives of the likelihood function are large or the inner product between the prior and the likelihood function is small then the errors incurred by updating can be potentially large.  These errors can be combated by making $\Delta x$ small.  This is potentially expensive since $\Delta x = N^{-D}$ where $N$ is the number of points in the mesh approximating the posterior.

\begin{corollary}
Given the likelihood function satisfies the assumptions of~\thm{errbd}, the number of qubits needed to represent the prior distribution using a uniform mesh of $[0,1]^D$ to sufficient precision to guarantee that the error in the posterior mean after an update is at most $\epsilon$ is bounded above by
$$
D \left\lceil\log_2\left(\max_E \frac{\langle P(E|x),P(x)\rangle^2+3{D} \Lambda}{\epsilon \langle P(E|x),P(x)\rangle^2}
\right) \right\rceil.
$$
\end{corollary}
\begin{proof}
Proof is an immediate consequence of substituting $\Delta x =1/N^{1/D}$ into~\thm{errbd} and solving for $N$.
\end{proof}
This shows that the number of qubits needed is at most logarithmic in the error.  Furthermore, if $L$ updates are required then the total cost is increased by at most an additive factor of $\log(L)$.  We do not include this in our cost estimates since this error estimate is needlessly pessimistic as Bayesian inference is insensitive to the initial prior according to the Bernstein von-Mises theorem and consequently such errors are unlikely to be additive.

\section{Filtering Distributions for Time-Dependent Models}
\label{app:time-dep}

In practice a physical system whose properties we want to infer is seldom time
independent. Moreover, by demanding time-invariance, we preclude applications
to many interesting problem domains outside of physics, such as in financial
modeling and computer vision. In such cases, the likelihood function $P(E|x)$
is replaced by $P(E|x;\tau)$ where $\tau$ is the time in the experimental
system. Thus, the techniques developed do not directly apply to time-dependent
cases, but rather to the model that results from marginalizing over this time-
dependence.

There are several ways of dealing with problems involving a time-dependent
likelihood. The most natural way is by introducing new parameters, called
\emph{hyperparameters}, that allow the variation of the likelihood function to
be modeled. Estimation and inference then proceed on the hyperparameters,
rather than on the latent variables directly. For instance, letting $\omega$
in the periodic likelihood 
\begin{align}
P(1 | \omega; \omega_-, t) &= \cos^2 ((\omega-\omega_-)  t),\nonumber\\
P(0 | \omega; \omega_-, t) &= \sin^2 ((\omega-\omega_-)  t),\label{eq:simple-precession}
\end{align}
 be drawn from a stationary Gaussian
process and then marginalizing over the history of that process results in a
new hyperparameterized likelihood
\begin{equation}
    P(1 | \mu, \sigma; \omega_-, t) = \frac{1}{2} \left(e^{-\frac{1}{2} \sigma ^2 t^2} \cos (\mu  t)+1\right),
\end{equation}
where $\mu$ and $\sigma^2$ are the mean and variance of the Gaussian process.
Using hyperparameters works well for modeling the distribution of the
dynamics of a model and can be directly implemented using the previously
discussed methods, but it does little to help track the \emph{instantaneous}
latent variables of the system.

Tracking such variation in the latent variables can be 
challenging because as Bayesian inference proceeds the certainty in the
value of the latent variables tends to increase, but if these variables drift beyond the support of the posterior distribution
then Bayesian inference will never be able to recover. 
In other words, if the true model for a system drifts into a region that the is not supported by the a prior obtained by previous updates that neglected
the stochasticity of the latent variables then the algorithm will no longer be able to track the instantaneous value of the latent variable.


Fortunately, the SMC literature has already provided a solution to this problem.  Approximate inference algorithms can
be made to track stochastically varying latent variables, by incorporating a prediction step that diffuses the hidden variables
of each particle \cite{isard_condensationconditional_1998}. Here, we extend this technique to the our
quantum algorithm by performing Bayes updates on QFT-transformed posterior
states. This allows our algorithm to continue enjoying dramatic advantages
in space complexity even in the presence of time-dependence.

In particular, by convolving the prior with a filter function such as a
Gaussian, the width of the resultant distribution can be increased without
affecting the prior mean.  This means that the Bayes estimate of the true
model will remain identical while granting the prior the ability to recover
from time-variation of the latent variables.  In particular, if we assume
that at each step Bayesian inference causes the posterior variance to contract
by a factor of $\alpha$ and  convolution with a filter function causes the
variance to expand by $\beta$ then the viariance of the resulting distribution
asymptotes to $\beta/(1-\alpha)$. Thus we can combat $\sigma(x)$ from becoming
unrealistically small by applying such filtering strategies.

The convolution property of the Fourier transform gives for any two functions $P$ and $Q$
\begin{equation}
P \star Q \propto \mathcal{F}^{-1}\left(\mathcal{F}(P) \cdot \mathcal{F}(Q) \right),\label{eq:conv}
\end{equation}
where $\star$ is the circular convolution operation.
The quantum Fourier transform can therefore be used to convolve an unknown $P$ with a known distribution $Q$ that has an efficiently computable Fourier transform $\hat{Q}$.  This
convolution allows us to filter the prior distribution.

\begin{theorem}
Let  $O_{\hat{Q}}$ be a quantum oracle such that $O_{\hat{Q}}\ket{x} \ket{y} = \ket{x} \ket{y\oplus \sin^{-1}(\hat{Q} (x)/\Gamma_E)}$ where $\hat{Q}:=\mathcal{F}(Q)$ and $\hat{Q}(x)\le \Gamma_E$ and $Q(x) \in \mathbb{C}^{2^n}$.  Then given access to a unitary oracle $O_{\rm in}$ that prepares the state $\sum_x \sqrt{P(x)}\ket{x}$, the state $\sum_x \sqrt{(P\star Q) (x)}\ket{x}$ can be prepared with error $\epsilon$ in the $2$--norm using a number of queries that has an average-case query complexity of $O(\sqrt{\Gamma_E/ \langle \mathcal{F}(P), \mathcal{F}(Q)}\rangle)$.
\end{theorem}
\begin{proof}
Notice that Bayes updating the quantum SMC state consists of pointwise multiplication. As a result, applying \lem{rejection} in the Fourier domain, we can implement the convolution described above.  Doing so involves the following process
\begin{enumerate}
\item  Fourier transform the current posterior, $\ket{P}:=\sum_x P(x) \ket{x} \mapsto \hat{\mathcal{F}}\left(\sum_x P(x) \ket{x} \right):= \sum_k \omega_k \ket{k}$.
\item Prepare the Fourier-domain representation of the convolution kernel, \\$\sum_k \omega_k \ket{k}\mapsto \sum_k \omega_k \ket{k}\ket{\sin^{-1}(\sqrt{\hat{Q}(k)/\Gamma_E})}$.
\item Update by the convolution kernel and transform back, \\$\sum_k \omega_k \ket{k}\ket{\sin^{-1}(\sqrt{\hat{Q}(k)/\Gamma_E})}\mapsto \hat{\mathcal{F}}^{-1}\left(\sum_k \omega_k \ket{k}\ket{0}\left(\sqrt{\hat{Q}(k)/\Gamma_E}\ket{1}+\sqrt{1-\hat{Q}(k)/\Gamma_E}\ket{0} \right)\right)$.
\end{enumerate}
If $1$ is measured then the result will implement the circular convolution $P \star Q$ according to~\eq{conv} and Plancherel's theorem.

First, the query complexity of this algorithm is easy to estimate.  The initial state preparation requires a query to $O_{\rm in}$ and the calculation of $\hat{Q}(k)$ requires a query to $P_{\hat{Q}}$.  By using amplitude amplification on the $1$ result, we have that on average $O(\sqrt{\Gamma_E/ \langle \mathcal{F}(P), \mathcal{F}(Q)}\rangle)$ queries are required to prepare the state.
\end{proof}

\end{document}